\documentclass{article}
\pdfoutput=1
\usepackage{latexsym}
\usepackage{amssymb}
\usepackage{url}
\usepackage{colortbl}

\def\false{\mbox{\textit{false}}}
\def\iff{\mbox{ iff }}
\def\metaequiv{\iff}

\def\pre{\mbox{pre}}
\def\post{\mbox{post}}

\def\true{\mbox{\textit{true}}}

\newtheorem{definition}{Definition}[section]

\newtheorem{example}{Example}[section]

\newtheorem{proposition}[definition]{Proposition}
\newenvironment{proof}{\noindent\textbf{Proof}\itshape}{~\hfill~$\Box$ \newline\newline}
\newtheorem{corollary}[definition]{Corollary}

\title{On the difference between checking\\ integrity constraints before or after updates}

\begin{document}
\author{Davide Martinenghi\\\\
Politecnico di Milano\\
Piazza Leonardo da Vinci 32\\
20133 Milano, Italy\\
\url{davide.martinenghi@polimi.it}}
\date{}

\maketitle

\begin{abstract}
Integrity checking is a crucial issue, as databases change their instance all the time and therefore need to be checked continuously and rapidly.
Decades of research have produced a plethora of methods for checking integrity constraints of a database in an incremental manner.
However, not much has been said about \emph{when} to check integrity.
In this paper, we study the differences and similarities between checking integrity before an update (a.k.a. pre-test) or after (a.k.a. post-test) in order to assess the respective convenience and properties.
\end{abstract}

\section{Introduction}\label{sec:intro}

Integrity checking has been 
a perennial topic in almost all database conferences, journals and
research labs. Very large quantities of research activities and
publications testify to the importance of the issue. The motivation
which has stimulated these investigations is that integrity checking
is practically unfeasible for significant amounts of stored data
without a dedicated approach to optimize the process.
Progress has been made with extensions of basic approaches to
deductive, object-relational, XML-based, distributed and other kinds
of advanced database technology, as surveyed in \cite{MCD:IGP2006}. However,
what has not changed much are the fundamental ideas that are already
present in the seminal paper \cite{Nic82}.

The basic idea is that, in most cases, a simplified form of the set
of integrity constraints imposed on the database can be obtained
from a given update (or just an update schema) and the current state
of the database (or just the database schema). Thus, integrity,
which is supposed to be an invariant of all possible database
states,
is checked upon each update request, which in turn becomes effective
if the check yields that integrity is not violated. Here,
``simplified'' essentially means more efficiently evaluated at update
time.
Of course, efficiency is not unequivocally measurable. However, the
number of stored facts to be retrieved for constraint evaluation is
a good rule of thumb. Another such rule is the number of literals in
the simplified form. Also, the minimization (reduction or avoidance)
of repair costs
after having detected integrity violation is a factor to be
considered when assessing the efficiency of integrity checking.
To establish the new state is usually considered less costly than to
undo it, but for concurrent transactions, more so in distributed and
even more so in replicated databases, establishing the new state is
a sizable and non-negligible cost factor (think of the concurrency
control, management and communication rounds needed for distribution
and different replication strategies), while roll-backs are less of
a problem since they can be taken care of and optimized by the DBMS,
making use of its transaction logs.

Integrity checking methods may differ in some of the assumptions.
For instance, in \cite{Nic82} and in most of the publications on the same
subject that came after it, a categorical premise for the
correctness of the simplification approach has been that the
constraints to be checked 
for a given update $U$ are supposed to be satisfied in the ``old"
state, i.e., the database state given when $U$ is requested.
This assumption has been relaxed in~\cite{MT:TKDE2011}, thereby introducing the notion of inconsistency-tolerant integrity checking.
Other different assumptions may regard the class of integrity constraints and updates that are supported by the methods, as surveyed in~\cite{MCD:IGP2006}.

Unlike previous surveys, in this paper we study the main differences between approaches that require to check integrity after all updates have been applied and those that check integrity before the updates.
We shall see that these checks are generally non-interchangeable, unless the updates are of a specific kind.

\section{Preliminaries}
\label{sec:integrity-constraints-simplified}

In this section, we recapitulate, classify and discuss the main
characteristics of the simplification approach to integrity
checking. We fix basic
definitions and terminology, and
thus the framework into which our formalizations in the remainder
are cast. In Section \ref{subsec:simplifications} we introduce
abstract notions of various classes of simplifications and discuss
them in Section \ref{subsec:pre-tests-vs-post-tests}.

Throughout, we refer to the relational framework of deductive
databases,
i.e., relational databases with possibly recursive view definitions
described in clause form \cite{AHV95,KH:VLDBJ94}.
Thus, a {\em database} consists of a set of \emph{facts}
and a set of \emph{rules}, i.e.,
\emph{tuples} and \emph{views}, 
in the terminology of the relational model.

An {\em integrity constraint} expresses a semantic invariant, i.e.,
a condition supposed to hold in each state of the database. In
general, it can be expressed by any closed first-order logic formula
in the language of the database on which it is imposed.
Two kinds of normal form representations of integrity constraints
which both incur no loss of generality are prominent in the
literature: prenex normal form and denial form. The former, as
defined and used, e.g., in \cite{Nic82, LST87}, has all quantifiers
moved leftmost and all negation symbols moved innermost, by
equivalent rewritings of the original formula. The latter, as
defined and used, e.g., in \cite{SK88,KSS87}, has the form of
datalog clauses with empty head, expressing that, if their condition
is satisfied, then integrity is violated, and may need dedicated
view definitions to define these conditions by recurring on database
facts.
An 
\emph{integrity theory} 
is a finite set of integrity constraints, 
to be thought of as being imposed on some database.

For simplicity, we limit ourselves to databases that have a
unique standard model and no ``unknown" facts (i.e., each fact is
either true or false in the model), e.g., stratified databases,
and assume that database semantics is defined 
by this model.

For a closed formula $W$, we write $D\models W$ (resp.,
$D\not\models W$) to indicate that $W$ evaluates to $\true$ (resp.,
$\false$) in $D$'s standard model.
For a set of formulas $\Gamma$, we write $D\models \Gamma$ (resp.,
$D\not\models \Gamma$) to indicate that for every (resp., some)
formula $W\in\Gamma$ we have $D\models W$ (resp., $D\not\models W$).
If $W$ is an integrity constraint and $\Gamma$ an integrity theory,
it is usual to also say that $D$ {\em satisfies} (resp., {\em
violates}) $W$ and $\Gamma$, respectively. An equivalent terminology
is defined as follows.
\begin{definition}[Consistency]
A database $D$ is {\em consistent} with an integrity theory $\Gamma$
iff $D\models\Gamma$.
\end{definition}
Informally, we have already spoken of database states. More
formally, database states are determined by atomically executed
updates. An \emph{update} $U$ is a mapping $U:{\cal D} \mapsto {\cal
D}$, where ${\cal D}$ is the space of databases
as determined by a
fixed, sufficiently rich language which need 
not be extended by any update. For simplicity, we only consider
updates that may involve facts and rules in this paper, but no
integrity constraints, which are thought of as immutable.
For convenience, for any database $D$,
let $D^U$ denote the new database 
obtained by 
applying update $U$ on $D$. 

\section{Kinds of simplifications}
\label{subsec:simplifications}

Traditionally,
the integrity checking problem asks, given a set of integrity
constraints $\Gamma$, a database $D$ consistent with $\Gamma$, and
an update $U$,
whether $D^U\models\Gamma$ holds, i.e., whether the new database is
consistent with (i.e., satisfies) the integrity constraints.
However,
evaluating $\Gamma$ in $D^U$ may be prohibitively expensive. 
So, a reformulation of the problem is called for, trying to take
advantage of the incrementality of updates. Traditionally, all such
reformulations have been made
under the assumption that the old state 
is consistent.

We will discuss two kinds of 
such reformulations 
that have commonly been 
dealt with in the literature.
Both determine an alternative integrity theory $\Upsilon$ (which by
itself is later called a simplification), the evaluation of which is
supposed to be simpler than to evaluate $\Gamma$, while yielding
equivalent results.
The first kind 
of such $\Upsilon$
is determined to be evaluated in the new 
state. 
\begin{definition}[Post-test]\label{def:post-test}
Let $\Gamma$ be an integrity theory and $U$ an update. An integrity
theory $\Upsilon$ is a \emph{post-test} of $\Gamma$ for $U$ whenever
$D^U\models\Gamma\mbox{ iff } D^U\models \Upsilon$ for 
every
database $D$ consistent with $\Gamma$.
\end{definition}
Clearly, $\Gamma$ itself is a post-test of $\Gamma$
for 
any update.
As indicated, one is interested in producing a post-test that is
actually ``simpler'' to evaluate than the original integrity
constraints. This
is traditionally 
achieved by exploiting
the fact 
that the old 
state $D$ is consistent with $\Gamma$, thus 
avoiding redundant checks
of cases that are already known to satisfy integrity.
Note that the process of integrity checking involving a post-test
consists in: executing the update, checking the post-test and, if it
fails to hold, correcting the database by
performing a repair action, i.e., a rollback 
and optionally a modification of the update
which won't violate integrity. 
Well-known post-test-based approaches are described in
\cite{Nic82,LT85,GM90,DBLP:conf/vldb/Olive91,LD98}.

The second kind of approach to deal with integrity checking
incrementally is to determine an integrity theory $\Sigma$
to be evaluated in the old state, i.e., to predict without actually
executing the update whether the new, updated state will be
consistent with the integrity constraints.
\begin{definition}[Pre-test]\label{def:pre-test}
Let $\Gamma$ be an integrity theory and $U$ an update. An integrity
theory $\Sigma$ is a \emph{pre-test} of $\Gamma$ for $U$ whenever
$D^U\models\Gamma\mbox{ iff } D\models \Sigma$ for 
every
database $D$ consistent with $\Gamma$.
\end{definition}
%
Here, not only the consistency of the old state $D$ with $\Gamma$ is
exploited, but also that inconsistent states can be prevented
without executing the update, and, thus, without ever
having to undo a violated new state. 
The integrity checking process involving a pre-test is therefore:
check whether the pre-test holds and, if it does, execute the
update. Examples of pre-test-based approaches are
\cite{Qia88,CM:FI2006}.
Note that, depending on the requirements of availability and
consistency of a given application, integrity checking with
pre-tests is possibly better suited for concurrent transaction
processing, particularly in distributed databases.

In the remainder, we refer to both post- 
and pre-tests 
as \emph{simplifications} of the original integrity theory.
It is tacitly assumed that
given simplifications
are, at least in significant classes of cases, indeed simpler to
evaluate than the original constraints. A characterization of
simplicity beyond what is mentioned in the introduction, e.g., by
formal cost models,
is out of the scope of this paper;
cf.~\cite{M:PHD2005} 
for a discussion.

Now, whatever the definition of simplicity, 
we cannot directly compare the evaluation cost
of a post-test 
with that of a pre-test 
since they refer to two different (viz. old and new) states.
Therefore,
it is desirable to have kindred 
reference pre- and post-tests for benchmarking given
simplifications.
Plain 
tests, as defined below, i.e., simplifications that do not exploit
the fact that the old state satisfies integrity,
may serve as such reference tests.
\begin{definition}[Plain test]\label{def:plain-test}
Let $\Gamma$ be an integrity theory and $U$ an update.

a) An integrity theory $\Sigma_0$ is a \emph{plain pre-test} of
$\Gamma$ for
$U$, denoted by 
$\pre_0^U(\Gamma)$, if the following holds: $D\models\Sigma_0 \mbox{
iff } D^U\models\Gamma$ for every database $D$.

b) An integrity theory $\Upsilon_0$ is a \emph{plain post-test} of
$\Gamma$ for $U$, denoted by $\post_0^U(\Gamma)$, if the following
holds: $D^U\models\Upsilon_0 \mbox{ iff } D^U\models\Gamma$  for
every database $D$.
\end{definition}
Clearly, $\Gamma$ is a plain 
post-test
 of itself
for 
any update.
Note that each 
plain test is also a simplification.
For any pre-test (resp., post-test),
it is therefore desirable that it be at least as simple
to evaluate 
as the corresponding plain test. 

It is straightforward to see that, for fixed $\Gamma$ and $U$, all
plain 
pre-tests of $\Gamma$ for 
$U$ are logically equivalent.
\begin{proposition}\label{pro:all-plain-pre-tests-equivalent}
Let $\Gamma$ be an integrity theory and $U$ an update. Then, for any
two 
plain
pre-tests $\Sigma'$ and $\Sigma''$ of $\Gamma$ for
$U$, we have $\Sigma'\equiv\Sigma''$.
\end{proposition}
\begin{proof}
By applying definition \ref{def:plain-test} to $\Sigma'$ and
$\Sigma''$, one gets by transitivity that $D\models\Sigma'\iff
D\models\Sigma''$ for 
every
$D$, i.e., $\Sigma'\equiv\Sigma''$.
\end{proof}
Conversely, not all plain post-tests are logically equivalent.
As a counterexample,
take, e.g., $\Gamma=\{\leftarrow p(a)\}$ and let $U$ be the
insertion of $p(a)$. Then $\Upsilon=\false$ is a post-test of
$\Gamma$ for 
$U$ but 
$\Upsilon\not\equiv\Gamma$.

We conclude this section with an example of simplification of
integrity constraints.
\begin{example}\label{ex:conflict-of-interests}
Consider a database with the relations $rev(S,R)$ (submission $S$
assigned to reviewer $R$), $sub(S,A)$ (submission $S$ authored by
$A$) and $pub(P,A)$ (publication $P$ authored by $A$). Assume a
policy establishing that no one can review a paper of his/her
(former) coauthors. This is expressed by:
$$\begin{array}{rl}
\!\!\Gamma\!\!=\!\!\{&\!\!\!\!\leftarrow \!\!rev(S,R)\land sub(S,R),\\
\!\!& \!\!\!\!\leftarrow \!\!rev(S,R)\land sub(S,A)\land
pub(P,R)\land pub(P,A)\}
\end{array}$$
Let $U$ be the an update that inserts the facts $sub(c,a)$ and
$rev(c,b)$ into the database, where $a$, $b$, $c$ are some
constants.
A simplification of $\Gamma$ for $U$ (equivalent to what Nicolas'
method would output) is as follows:
$$\begin{array}{rl}
\Sigma=\{&\leftarrow sub(c,b),\\
& \leftarrow rev(c,a),\\
& \leftarrow pub(P,b)\land pub(P,a),\\
& \leftarrow sub(c,A)\land pub(P,b)\land pub(P,A),\\
& \leftarrow rev(c,R)\land pub(P,R)\land pub(P,a)\;\;\}
\end{array}$$
The simplified conditions given by $\Sigma$ can be read as follows:
\begin{itemize}
\item $b$ did not submit $c$
\item $a$ does not review $c$
\item $b$ is not coauthor of $a$
\item $b$ is not coauthor of an author of $c$
\item $c$ is not reviewed by a coauthor of $a$
\end{itemize}
These checks are much easier to execute than $\Gamma$, as they
greatly reduce the space of tuples to be considered by virtue of the
instantiation of variables with constants.
\end{example}

\section{Relationship between pre- and post-tests}
\label{subsec:pre-tests-vs-post-tests}
In this section we compare pre-tests and post-tests and discuss
their interchangeability.
Note that 
we do this without referring to any specific simplification method
or update language.

First, we show that, in general, a pre-test cannot be used as a
post-test, nor vice versa.

\begin{example}\label{ex:pre-and-post-test-differ-in-general}
Consider the integrity theory $\Gamma=\{\leftarrow p(a)\land q(b)\}$
and an update $U$ that exchanges the contents of $p$ and $q$. Then
$\Sigma=\{\leftarrow q(a)\land p(b)\}$ is a pre-test but clearly not
a correct post-test. Consider, e.g., a database
$D=\{p(a),p(b),q(a)\}$; we have $D\models\Gamma,
D^U\not\models\Gamma, D^U\models\Sigma$, i.e., it does not hold that
$D^U\models\Sigma \metaequiv D^U\models\Gamma$, although $D$ is
consistent with $\Gamma$. Similarly, $\Upsilon=\{\leftarrow
p(a)\land q(b)\}$ is a post-test, but not a pre-test, of $\Gamma$
for $U$.
\end{example}
This result is not surprising, since we allow for updates
representing any kind of transformation of the database, such as
swapping the contents of two relations.
We now introduce a class of updates that excludes an update such as
$U$ of example \ref{ex:pre-and-post-test-differ-in-general}.
\begin{definition}\label{def:idempotent-update}
An update $U$ is \emph{idempotent}\index{update!idempotent update}
if $D^U$ = $(D^U)^U$ for any database $D$.
\end{definition}
Idempotent updates capture additions, deletions and changes of
specific tuples, which are certainly among the most frequent kinds
of updates.
For idempotent updates, we can prove that a plain pre-test is also
always a valid plain post-test.
\begin{proposition}\label{pro:plain-pre-test=post-test-idempotent}
Let $\Gamma$ be an integrity theory and $U$ an idempotent update.
Then $D^U\models\pre_0^U(\Gamma)\metaequiv D^U\models\Gamma$ for any
database $D$, i.e., $pre_0^U(\Gamma)$ is a plain post-test of
$\Gamma$ for $U$.
\end{proposition}
\begin{proof}
Since $U$ is idempotent, i.e., $D^U = (D^U)^U$ for any $D$, we have
\begin{itemize}
\item[(1)] $D^U\models\Gamma\metaequiv (D^U)^U\models\Gamma$ for any
$D$.
\end{itemize}
Since $\pre_0^U(\Gamma)$ is a plain pre-test of $\Gamma$ wrt. $U$,
we have
\begin{itemize}
\item[(2)] $D^U\models\pre_0^U(\Gamma)\metaequiv
(D^U)^U\models\Gamma$ for any $D^U$ (and thus for any $D$).
\end{itemize}
\noindent By transitivity between (1) and (2) we obtain the thesis.
\end{proof}
More surprisingly, however, the converse does not hold, i.e., there
are plain post-tests that are not plain pre-tests. In general, even
for idempotent updates, there are pre-tests that are not post-tests
and post-tests that are not pre-tests.
\begin{proposition}\label{pro:pre-differ-post-for-idempotent}
Let $\Gamma$ be an integrity theory and $U$ an idempotent update.
Then
\begin{enumerate}
\item [(1)] there is a pre-test $\Sigma$ of $\Gamma$ for $U$ such
that it does not hold that $D^U\models\Sigma\metaequiv
D^U\models\Gamma$ for any database $D$ consistent with $\Gamma$,
i.e., $\Sigma$ is not a post-test of $\Gamma$ for $U$.
\item [(2)] there is a plain post-test $\Upsilon$ of $\Gamma$ for
$U$ such that it does not hold that $D\models\Upsilon\metaequiv
D^U\models\Gamma$ for any database $D$ consistent with $\Gamma$,
i.e., $\Upsilon$ is not a pre-test of $\Gamma$ for $U$.
\end{enumerate}
\end{proposition}
\begin{proof}

\noindent (1) Let $D$ be a state consistent with $\Gamma$. We have:
$$\begin{array}{rll}
D\models\Sigma&\metaequiv D^U\models \Gamma \mbox{ ($\Sigma$ pre-test)}\\
&\metaequiv (D^U)^U\models \Gamma \mbox{ ($U$ idempotent)
}\\
&\metaequiv D^U\models \Sigma, \mbox{ if $D^U\models\Gamma$
($\Sigma$ pre-test) }
\end{array}$$
So, the only possibility is a situation where $D\not\models\Sigma$,
$D^U\models\Sigma$, and $D^U\not\models\Gamma$, which happens, e.g.,
with $\Gamma = \{\leftarrow p(a)\}$, $\Sigma = \{\leftarrow
\sim$$p(a)\}$, $D=\emptyset$ and $U$ an update such that
$D^U=\{p(a)\}$.
To conclude the proof, we show that the chosen $\Sigma$ is a
pre-test. Indeed, for any $D$ such that $D\models\Gamma$, then
$D\not\models\Sigma$, and $D^U\not\models\Gamma$.

\noindent(2) The same $D$, $U$ and $\Gamma$ as in the previous point
can also be used for the second case by considering the plain
post-test $\Upsilon=\Gamma$. We have $D\models\Gamma$,
$D^U\not\models\Gamma$, $D\models\Upsilon$ and
$D^U\not\models\Upsilon$.
\end{proof}
Another interesting aspect regarding pre-tests and post-tests is
whether their evaluation is at all affected by the update.
Proposition \ref{pro:plain-pre-test=post-test-idempotent}
immediately implies that the evaluation of a plain pre-test is not
affected by the update, as stated in the following corollary.
\begin{corollary}\label{cor:plain-pre-test-resource-set-update-disjoint}
Let $\Gamma$ be an integrity theory and $U$ an idempotent update.
Then $D\models\pre_0^U(\Gamma)\metaequiv D^U\models\pre_0^U(\Gamma)$
for any $D$.
\end{corollary}
\begin{proof}
By definition of $\pre_0^U(\Gamma)$, we have
\begin{itemize}
\item[] $D\models\pre_0^U(\Gamma)\metaequiv D^U\models\Gamma$ for
any $D$,
\end{itemize}
and, by transitivity with the claim of proposition
\ref{pro:plain-pre-test=post-test-idempotent}, we have the thesis.
\end{proof}
However, this does not hold in general for pre-tests or (plain)
post-tests, as demonstrated in the example of the proof of
proposition \ref{pro:pre-differ-post-for-idempotent}.

The following table summarizes the results presented in this
section. We indicate with Pre$^U(\Gamma)$ (resp., Post$^U(\Gamma)$)
the set of all pre-tests (resp., post-tests) of integrity theory
$\Gamma$ for $U$, and use a $0$ subscript to indicate the set of all
plain pre-tests (resp., post-tests) of $\Gamma$ for $U$.
\vspace{0.2cm}

\setlength{\extrarowheight}{0.05cm}
\noindent\begin{center}\begin{tabular}{|l|c|c|} \hline
for any $\Gamma$ & any $U$ & $U$ idempotent\\
\hline
Pre$^U(\Gamma)\subseteq$Post$^U(\Gamma)$? & \textbf{no} & \textbf{no}\\
\hline
Post$^U(\Gamma)\subseteq$Pre$^U(\Gamma)$? & \textbf{no} & \textbf{no}\\
\hline
Pre$_0^U(\Gamma)\subseteq$Pre$^U(\Gamma)$? & \textbf{yes} & \textbf{yes}\\
\hline
Post$_0^U(\Gamma)\subseteq$Post$^U(\Gamma)$? & \textbf{yes} & \textbf{yes}\\
\hline
Pre$_0^U(\Gamma)\subseteq$Post$^U(\Gamma)$? & \textbf{no} & \textbf{yes}\\
\hline
\end{tabular}\end{center}
\setlength{\extrarowheight}{0cm}

\section{Related work}\label{sec:related}
%
Simplification of integrity constraints has been recognized by a
large body of research as a powerful technique for optimization of
integrity checking.
Several approaches to simplification 
require the update transaction to be performed
{\em before} the resulting state is 
checked for consistency with a
post-test~\cite{Nic82,LST87,SK88,GM90,DC94,LD98}.
%
%
As opposed to that,
to pre-test the feasibility of an update with respect to an
integrity theory allows for avoiding both the execution of the
update and, particularly, the restoration of the database state
before the update, which may be very costly.
%
Pre-test-based methods are, e.g.,
\cite{HMN84,DBLP:conf/sigmod/HsuI85,Qia88,LL96,LD98,CM:FI2006,DM:IGP2008,MCD:IGP2006,CM:LPAR2005,MC:ADBIS2005,MC:DEXA2005,M:ADBIS2004,M:FQAS2004,CM:LOPSTR2003,M:PHD2005,martinenghi2003simplification,CM:FoIKS2004,CM:LAAIC2006},
including a few industrial attempts, e.g.,
\cite{Carrico95,Benedikt02}.
Other methods provide simplifications that may require the
availability of both the old and the new state, assuming that the
database keeps track of the old state before committing an update,
\cite{S95,SS99}.

In \cite{GM90}, an adaptation of subsumption checking (called
\emph{partial subsumption}) is used to generate simplification as
the ``difference'' (called \emph{residual}) between an integrity
constraint and a clause representing an update.

Qian's method~\cite{Qia88} generates pre-tests for integrity
checking based on the observation that a simplified integrity
constraint can be regarded as a weakest precondition for having a
consistent updated state, in the same sense as in Hoare's
logic~\cite{H69,D76} for imperative languages, and by assuming
consistency of the database before the update.

Simplification of integrity constraints for update patterns
resembles the notion of program \emph{specialization} used in
partial evaluation, which is the process of creating a specialized
version of a given program (in this case, a general integrity
checker) with respect to known input data (here, the update), as
investigated in \cite{LD98}.

More generally, integrity checking can be seen as a special case of
materialized view maintenance: integrity constraints are defined as
views that must always remain empty for the database to be
consistent~\cite{310709,DBLP:journals/sigmod/DongS00}.

Simplification can also be obtained by resorting to decision
procedures for query
containment~\cite{Florescu:1998:QCC,Calvanese:1998:DQC,CDV03,DBLP:conf/pods/Bonatti04},
as shown in \cite{CM:FI2006}.

We intentionally did not do so before, but at this point is seems
worth mentioning that several simplification methods accept
instantiable or parameterizable {\em patterns} of
updates instead of specific updates, e.g. 
~\cite{HMN84,GM90,CM:FI2006}. Thus, given such a pattern at schema
specification time, it is possible to compile a simplification of
the integrity theory for all updates matching that pattern. For
instance, if constants $a$, $b$ and $c$ in the update of Example
\ref{ex:conflict-of-interests} were specified as simple placeholders
for constants (called \emph{dummy constants} in \cite{HMN84} and
\emph{parameters} in \cite{CM:FI2006}), and thus not assumed to be
necessarily different, a pre-test-based simplification would also
include an integrity constraint that checks that $a\neq b$.

Logic programming-based approaches such as \cite{SK88,KSS87} do not
take into account irrelevant clauses for refuting denial
constraints, even if they would take part in an unnoticed case of
inconsistency that has not been caused by the checked update but by
some earlier event. Moreover, such approaches do not exhibit any
explosive behavior as predicted classical logic in the presence of
inconsistency. In other words, query evaluation procedures based on
SL-resolution can fairly well be called inconsistency-tolerant or
``paraconsistent" in a procedural sense, as done, e.g., in
\cite{kowalski_book,lncs_2582}.
The declarative inconsistency
tolerance of simplifications for improving integrity checking has been studied in several works~\cite{DM:TKDE2011,DM:PPDP2008,DM:LPAR2006,DM:ADVIS2006,DM:SEBD2006,DM:QOIS2009,DM:FlexDBIST2007,DM:TDM2006,DM:FlexDBIST2006,DM:IGP2009}.

The related problem of restoring integrity of a database once
inconsistencies are discovered is tackled by calculating a
\emph{repair}, i.e., a consistent database that is as close as
possible to the original, inconsistent database. Since the seminal
contribution \cite{ABC00}, many authors have studied the problem of
providing consistent answers to queries posed to an inconsistent
database. These techniques certainly add to inconsistency tolerance
in databases, but cannot be directly used to detect inconsistencies
for integrity checking purposes (i.e., by posing integrity
constraints as queries).
Along the same lines, active rules have been considered as a means
to restore a consistent database
\cite{DBLP:conf/vldb/CeriW90,DBLP:journals/tods/CeriFPT94,Dec02}.

Last, we mention work on incomplete databases which also considers
integrity constraints that are not satisfied in a given database
state as something to be dealt with constructively, instead of
banning it from consideration, as most integrity checking methods do
\cite{Vardi86, vanMeyden_survey,CM:TPLP2010,M:CORR2013b}. However, that work is not
interested in integrity checking simplifications that could be used
in such databases. Rather, it is dealing with the issue of
satisfiability and its computational complexity, as related to an
open world assumption by which the space of possible ``closed
worlds" (i.e. database states without missing information) that
would satisfy integrity are studied. The theme of this paper is to
simplify the checking of integrity satisfaction in the presence of
inconsistency, not to ask for the satisfiability of integrity
constraints in the absence of complete information. (Basic
similarities and differences of satisfaction and satisfiability of
integrity are addressed in \cite{bryDeckerManthey88}.)

Relevant new directions of research regard all those areas where integrity constraints are used to characterize useful scenarios in which query answering plays an important role.
Among these, we mention
\emph{i)} \emph{access patterns}, which are constraints indicating which attributes of a relation schema are used as input and which ones are used as output~\cite{CM:APWEB2010,CM:EDBT2010,M:CRYPT2011,CM:ER2008,CM:ICDE2008,CCM:EROW2007,CMC:SEBD2007,CCM:JUCS2009},
\emph{ii)} top-$K$ queries, where the constraints specify a limit on the number of results that the query should return, including constraints on proximity or diversity~\cite{MT:TODS2012,MT:PVLDB2010,CCFMT:TODS2013,FMT:SIGMOD2012,DBLP:conf/sebd/FraternaliMT12,MT:TKDE2011,CMT:TECHREP2009,MT:DBRANK2010,SIMT:SIGMOD2011,IMT:SECO2009},
\emph{iii)} \emph{taxonomies} and context information, which may be used to pose constraint on the granularity of the data and to reason about it~\cite{MT:ER2010,MT:FQAS2009,MT:ITAIS2010,MT:ITAIS2009,MT:TECHREP2009,RMT:LID2011,CM:AAI2000}.

\section{Conclusion}\label{sec:conclusion}

We have discussed and compared the two main abstract families of methods that can be used to incrementally check integrity constraints: pre-tests and post-tests.
These are simplifications to be checked before or, respectively, after the update is executed (while update commitment is supposed to occur only after a successful check).
In order to not only talk about some selected, specific methods, albeit well-known ones, we have characterized declarative and procedural aspects of simplification-based integrity checking in a manner which is largely independent of concrete methods.
Unsurprisingly, pre-tests and post-tests are not interchangeable, not only in terms of the convenience of executing the ones or the others in practical situations (pre-tests may actually be preferred in case of updates to be rejected), but also of their semantic properties.
Somewhat surprisingly, however, their applicability is mostly asymmetric, even for the simple case of idempotent updates.

\bibliographystyle{abbrv}
{\small \bibliography{DM_VLDB06}}
\end{document}